\documentclass[10pt, doublecolumn]{IEEEtran}
\usepackage{epsfig,latexsym}
\usepackage{float}
\usepackage{indentfirst}
\usepackage{amsmath}
\usepackage{amssymb}
\usepackage{times}
\usepackage{subfigure}
\usepackage{psfrag}
\usepackage{hyperref}
\usepackage{cite}
\usepackage{lastpage}
\usepackage{fancyhdr}
\usepackage{color}
 \usepackage{amsthm}
\usepackage{bigints}
\newtheorem{Proposition}{Proposition}
\newtheorem{Lemma}{Lemma}
\newtheorem{lemma}[Lemma]{$\mathbf{Lemma}$}
\newtheorem{proposition}[Proposition]{Proposition}

\newcounter{problem}
\newcounter{save@equation}
\newcounter{save@problem}
\makeatletter

\begin{document}%%
\title{ {\huge On the Impact of   Phase Shifting Designs on IRS-NOMA }}

\author{ Zhiguo Ding, \IEEEmembership{Fellow, IEEE}, Robert Schober, \IEEEmembership{Fellow, IEEE}, and H. Vincent Poor, \IEEEmembership{Fellow, IEEE}    \thanks{ 
  
\vspace{-2em}

    Z. Ding and H. V. Poor are  with the Department of
Electrical Engineering, Princeton University, Princeton, NJ 08544,
USA. Z. Ding
 is also  with the School of
Electrical and Electronic Engineering, the University of Manchester, Manchester, UK (email: \href{mailto:zhiguo.ding@manchester.ac.uk}{zhiguo.ding@manchester.ac.uk}, \href{mailto:poor@princeton.edu}{poor@princeton.edu}).
R. Schober is with the Institute for Digital Communications,
Friedrich-Alexander-University Erlangen-Nurnberg (FAU), Germany (email: \href{mailto:robert.schober@fau.de}{robert.schober@fau.de}).

  }\vspace{-2em}}
 \maketitle
 
\begin{abstract}  
In this letter, the impact of two phase shifting designs, namely random phase shifting and coherent phase shifting, on the performance of intelligent reflecting surface (IRS) assisted non-orthogonal multiple access (NOMA) is studied. Analytical results are developed to show that the two designs achieve different tradeoffs between  reliability  and complexity.  Simulation results are provided to    compare IRS-NOMA to conventional relaying and IRS assisted orthogonal multiple access, and also to verify the accuracy of   the obtained analytical results. 
\end{abstract} \vspace{-1em} 

\section{Introduction}
Recently,   intelligent reflecting surfaces (IRS) have received significant attention, due to their  ability to intelligently reconfigure   wireless communication environments for better reception reliability at a low cost \cite{irs2, irs1,irsce3,robertirsmy1}. Similar to finite-resolution analogue-beamforming (FRAB),   each reflecting element on the IRS changes the phase of the reflected signal only, without modifying   its amplitude   \cite{7918554}. As a new energy and spectrally efficient technique,   IRS have been recently shown   compatible with various advanced communication techniques, including  millimeter-wave communications, unmanned aerial vehicle networks, physical layer security, simultaneous wireless information and power transfer \cite{sk8e,skyx1,skyx2}. 

In this letter, we  investigate how the  phase shifting design affects the performance of IRS assisted non-orthogonal multiple access (NOMA).  Two types of phase shifting designs are considered. The first one is coherent phase shifting, where    the phase shift of each reflecting element is matched with the phases of its incoming and outgoing fading channels.   
Despite its superior performance,   coherent phase shifting might not be applicable in practice  because of the finite resolution of practical IRS phase shifters   and the excessive system overhead caused by  acquiring channel state information (CSI) at the source. This motivates the second design employing    random phase shifting, where the scheme developed in \cite{irsmy1} can be viewed as a special case. The central limit theorem (CLT) is shown to be an accurate approximation tool for analyzing  the performance of the random phase shifting scheme.  In contrast, for the coherent phase shifting scheme, the approximation  obtained with   the CLT is  accurate at   low signal-to-noise ratio (SNR)  only. This   motivates the development of an upper bound on the outage performance, which is shown to be more accurate than the CLT-based result in the high SNR regime.  

\vspace{-1em}
%%%%%%%%%%%%%%%%%%
\section{System Model}
Consider a cooperative communication scenario with one source and two users, denoted by $\text{U}_1$ and $\text{U}_2$, respectively.   We assume that a direct link is not available between the source and $\text{U}_1$  due to severe blockage. Three cooperative communication strategies are described in the following subsections, respectively. 

\subsubsection{IRS-NOMA}
IRS-NOMA  ensures that the two users are simultaneously served. In particular, the source broadcasts the superimposed message, $c_1s_1+c_2s_2$, where $s_i$ denotes the unit-power signal for ${\rm U}_i$, $c_i$ denotes the power allocation coefficient and we assume that $c_1\geq c_2$ and $c_1^2+c_2^2=1$. 

The signals received by  ${\rm U}_1$ and  ${\rm U}_2$ are given by
\begin{align}
y_1  =   \frac{\mathbf{g}_1^H\boldsymbol \Theta \mathbf{g}_0}{\sqrt{d_{r}^{\alpha}d^{\alpha}_{r1}}} \sqrt{P_s}\left(c_1s_1+c_2s_2\right) +w_1,
\end{align}
and
\begin{align}\label{y2d}
y_2  = \left(\frac{h_2}{\sqrt{d_2^{\alpha}}} + \frac{\mathbf{g}_2^H\boldsymbol \Theta \mathbf{g}_0}{\sqrt{d_{r}^{\alpha}d^{\alpha}_{r2}}}\right)\sqrt{P_s}\left(c_1s_1+c_2s_2\right) +w_2,
\end{align}
where $P_s$ denotes the power used by the source, $\boldsymbol \Theta$ denotes the $N\times N$ diagonal phase shifting matrix with its $N$ main diagonal elements representing the reflecting elements of the IRS, $\mathbf{g}_i$ denotes the fading vector between the IRS and ${\rm U}_i$,  $\mathbf{g}_0$ denotes the fading      vector  between the source and the IRS, $h_2$ and $d_2$ denote the Rayleigh fading gain and the distance from the source to $\text{U}_2$, respectively,  $d_r$ and $d_{ri}$ denote the distances from the IRS to the source and ${\rm U}_i$, respectively, $w_i$ denotes the    noise at ${\rm U}_i$, and $\alpha$ denotes the path loss exponent. We assume that the noise power is  normalized and all fading gains are complex Gaussian distributed with zero mean and unit variance, i.e., $\mathcal{CN}(0,1)$. 

${\rm U}_1$ treats $s_2$ as noise when decoding its own signal, $s_1$, which means that the outage probability achieved by IRS-NOMA is given by
\begin{align}\label{noma}
{\rm P}_1^{{\rm NOMA}} &= {\rm P}\left( \log\left(1+\frac{P_sc_1^2 \frac{ \left|\mathbf{g}_1^H\boldsymbol \Theta \mathbf{g}_0\right|^2}{ {d_{r}^{\alpha}d^{\alpha}_{r1}}}}{P_s c_2^2\frac{ \left|\mathbf{g}_1^H\boldsymbol \Theta \mathbf{g}_0\right|^2}{ {d_{r}^{\alpha}d^{\alpha}_{r1}}}+1}\right)< R_1 \right) ,
\end{align} 
where $R_i$ denotes $\text{U}_i$'s target data rate and ${\rm P}(E)$ denotes the probability of event $E$. 
We note that the reflecting path in \eqref{y2d} can be insignificantly weaker than the direct link. For example, for a case with $d_2=d_r$, $d_{r2}=10$ m and $\alpha=4$, the path loss of the reflecting path is $10^{4}$ times larger than  that of the direct link. For such scenarios,  ${\rm U}_2$'s outage probability is similar to that in conventional NOMA without IRS.

\subsubsection{Conventional Relaying}
A straightforward benchmarking scheme is  cooperative  OMA transmission without IRS, which   involves  two phases.    The first phase needs to be further divided into two time slots. During the first time slot,  the source sends   $s_1$ to $\text{U}_2$, and during the second time slot,   $\text{U}_2$ forwards $s_1$ to $\text{U}_1$, if it can decode $s_1$. Otherwise, $\text{U}_2$  keeps silent. During the second phase, the source serves $\text{U}_2$ directly. Therefore, the outage probabilities for $\text{U}_i$ are given by 
\begin{align} \nonumber
{\rm P}_1^{{\rm OMA}} &=   {\rm P}\left( E_1\right)+  {\rm P}\left(\frac{1}{4}\log\left(1+ P_r \frac{|h_{12}|^2}{d_{12}^{\alpha}}\right)< R_1, E_1^c\right), 
\end{align}
 and ${\rm P}_2^{{\rm OMA}} = {\rm P}\left(\frac{1}{2}\log\left(1+P_2 \frac{|h_2|^2}{d_2^{\alpha}}\right)< R_2\right),$ where $E_1\triangleq \left\{\frac{1}{4}\log\left(1+P_1 \frac{|h_2|^2}{d_2^{\alpha}}\right)\leq R_1\right\}$, respectively, $E_1^c$ denotes the complementary event of $E_1$,  $h_{12}$ and $d_{12}$ denote the fading gain and the distance between the two users, respectively,    $P_i$ denotes the source's transmit power for ${\rm U}_i$'s signal, and  $P_r$ denotes the relay transmission power.   For a fair comparison, we assume that $ P_1=P_r=P_2=P_s$. 

%We note that it is possible that conventional relaying outperforms IRS transmission, since  conventional relaying avoids the product of the path losses, ${d_{r}^{\alpha}d^{\alpha}_{ri}}$.  

\subsubsection{IRS-OMA}
IRS assisted cooperative OMA also includes two phases. During the $i$-th phase,   the source serves  $\text{U}_i$ with the help of the IRS.    
 Therefore, the outage probability experienced by    ${\rm U}_1$ is given by
\begin{align}\label{ioma}
{\rm P}_1^{{\rm I-OMA}} &= {\rm P}\left(\frac{1}{2}\log\left(1+P_0  \frac{\left| \mathbf{g}_1^H\boldsymbol \Theta \mathbf{g}_0\right|^2}{ {d_{r}^{\alpha}d^{\alpha}_{r1}}}\right)< R_1 \right), \end{align} 
where $ P_0$ denotes the transmit power of $s_1$ in IRS-OMA.  For a fair comparison, we assume that  $P_0=P_s$.  The outage probability for ${\rm U}_2$ can be obtained similarly.

\section{Performance Analysis}
It is straightforward to observe that ${\rm U}_2$'s outage performance in IRS-NOMA  is much better than those of the benchmarking schemes. Because the   analysis of ${\rm U}_2$'s outage probability is similar to the case without the IRS,   in the remainder of the letter, we mainly focus on the outage probability experienced by  ${\rm U}_1$. 

By assuming that the fading is Rayleigh distributed, it is straightforward to show that the outage probability achieved by conventional relaying is given by 
\begin{align}\label{comax}
{\rm P}_1^{{\rm OMA}} &=   \left(1-e^{-d_{2}^{\alpha}\left(2^{4R_1}-1\right)P_1^{-1}}\right)+  e^{-d_{2}^{\alpha}\left(2^{4R_1}-1\right)P_1^{-1}} \\\nonumber &\times  \left(1-e^{-d_{12}^{\alpha}\left(2^{4R_1}-1\right)P_1^{-1}}\right). 
\end{align}

The evaluation of the outage probabilities, ${\rm P}_1^{{\rm NOMA}} $ and ${\rm P}_1^{{\rm I-OMA}}$, depends on how the phase shifting matrix $\boldsymbol \Theta$ is designed, where two designs with different tradeoffs between performance and complexity are introduced in the following.  

\vspace{-0.5em}
\subsection{Coherent Phase Shifting} 
Define $\mathbf{g}_i^H\boldsymbol \Theta \mathbf{g}_0$ as follows:
\begin{align}\label{theta}
\xi_N \triangleq  \mathbf{g}_i^H\boldsymbol \Theta \mathbf{g}_0 = \sum^{N}_{n=1} e^{-j\theta_{n}} g_{0,n} g_{i,n},
\end{align}
where $g_{i,n}$ and $g_{0,n}$ denote the $n$-th elements of $\mathbf{g}_i$ and $\mathbf{g}_0$, respectively, and   $\theta_n$ denotes the   phase shift of the $n$-th reflecting element of the IRS.

Assume that the phase of $g_{0,n} g_{i,n}$ can be acquired by the source. For the coherent phase shifting design, the phase shifts of the IRS are matched with the phases of the IRS fading gains, which yields the following:  
\begin{align}
\xi_N   = \sum^{N}_{n=1} \left| g_{0,n} g_{i,n}\right|. 
\end{align}

The evaluation of the outage probabilities requires   knowledge of the probability density function (pdf) of $\xi_N$, which is difficult to obtain. Therefore, two approximations are considered  in the following subsections.

\subsubsection{CLT-based Approximation} We note that the $\left| g_{0,n} g_{i,n}\right|$, $1\leq n \leq N$, are independent and identically distributed (i.i.d.), and hence $\xi_N$ is a sum of i.i.d. random variables, which motivates the use of the CLT. 
 
As shown in \cite{irsmy1,6725568},  the pdf of $\left| g_{0,n} g_{i,n}\right|^2$ is given by
\begin{align}
f_{\left| g_{0,n} g_{i,n}\right|^2}(x) = 2K_0\left(2\sqrt{x}\right),
\end{align}
which can be used to find the pdf of $\left| g_{0,n} g_{i,n}\right|$ as follows:
\begin{align}\label{abs1}
f_{\left| g_{0,n} g_{i,n}\right|}(y) = 4yK_0\left(2y\right),
\end{align}
where $K_i(\cdot)$ denotes the modified Bessel function of the second kind.  

The use of the CLT requires  the mean and the variance of $\left| g_{0,n} g_{i,n}\right|$, which can be obtained as follows:
\begin{align}
\mu_{\left| g_{0,n} g_{i,n}\right|} = &\int^{\infty}_{0} 4y^2K_0\left(2y\right)dy=\frac{\pi}{4},
\end{align}
and
\begin{align}
\sigma^2_{\left| g_{0,n} g_{i,n}\right|} = &\int^{\infty}_{0} 4y^3K_0\left(2y\right)dy - \frac{\pi^2}{16}=1-\frac{\pi^2}{16},
\end{align}
which follows from \cite[(6.561.16)]{GRADSHTEYN}.
 By using the CLT,   $\xi_N$ can be approximated as the following Gaussian random variable: 
\begin{align}\label{gauss}
\sqrt{N}\left(\frac{\xi_N}{N}-\mu_{\left| g_{0,n} g_{i,n}\right|}\right) \sim \mathcal{N}\left(0, \sigma^2_{\left| g_{0,n} g_{i,n}\right|}\right).
\end{align} 

Therefore, the outage probability of IRS-NOMA can be approximated  as follows:
\begin{align}\nonumber
{\rm P}_1^{{\rm NOMA}} &= {\rm P}\left(  \mathbf{g}_1^H\boldsymbol \Theta \mathbf{g}_0  <\sqrt{\epsilon_1} \right)\\ \label{iomaxge}
&\approx\frac{1}{2}+\frac{1}{2}\phi\left(\frac{\sqrt{N}\left(\frac{\sqrt{\epsilon_1}}{N}-\mu_{\left| g_{0,n} g_{i,n}\right|}\right)}{\sqrt{2}  \sigma_{\left| g_{0,n} g_{i,n}\right|}}\right), 
\end{align} 
where $\epsilon_1= \frac{{d_{r}^{\alpha}d^{\alpha}_{r1}}\epsilon}{P_s(c_1^2-\epsilon c_2^2)}$, $\epsilon=2^{R_1}-1$,  and $\phi(x)\triangleq \frac{2}{\sqrt{\pi}}\int^{x}_{0}e^{-t^2}dt$. It is assumed in this letter that    $c_1^2>\epsilon c_2^2$, otherwise ${\rm P}_1^{{\rm NOMA}}=1$. The outage probability for IRS-OMA can be obtained similarly by replacing $\epsilon_1$ with $\epsilon_2 = \frac{{d_{r}^{\alpha}d^{\alpha}_{r1}}\left(2^{2R_1}-1\right)}{P_0}$ in \eqref{iomaxge}.

\subsubsection{An upper bound} As shown in Section \ref{section sim}, the CLT-based approximation is not accurate in the high SNR regime, which motivates the development of a tight bound for the outage probability. 
To this end, we   focus on IRS-NOMA. An upper bound on the outage probability is given  in the following lemma. 

\begin{lemma}\label{lemma0}
Assume that $N$ is an even number and denote $\bar{N}=\frac{N}{2}$. The outage probability achieved by IRS-NOMA is upper bounded as follows:
\begin{align} \label{iomax41s}
{\rm P}_1^{{\rm NOMA}} 
&\leq   \frac{2^N\pi^{\frac{N}{2}}\Gamma^N\left(\frac{3}{2}\right)}{ (3\bar{N}-1)!}  2^{-3\bar{N}} \gamma\left(3\bar{N}, 2\sqrt{\epsilon_1}\right),
\end{align}
where $\gamma(\cdot,\cdot)$ denotes the incomplete Gamma function and $\Gamma(\cdot)$ denotes the Gamma function.  
\end{lemma}
\begin{proof}
See Appendix A. 
\end{proof}
{\it Remark:} At high SNR, $\epsilon_1\rightarrow 0$, which yields the following approximation for the upper bound shown in \eqref{iomax41s}:
\begin{align} 
{\rm P}_1^{{\rm NOMA}} 
\leq  
\label{iomax42}        2^{N-3\bar{N}}\pi^{\frac{N}{2}}\Gamma^N\left(\frac{3}{2}\right)   \frac{ 2 ^{3\bar{N}}}{(3\bar{N})!}   \epsilon_1 ^{\frac{3}{4}{N}}\doteq \frac{1}{P_0^{\frac{3}{4}{N}}},
\end{align}
where   $\doteq$  denotes exponential equality  \cite{Zhengl03}. Eq. \eqref{iomax42} indicates that $\frac{3}{4}N$ is an achievable diversity order. It is important to point out that the full diversity gain, $N$, is also achievable, as shown in the following. By using the fact that $\xi_N\geq \left| g_{0,n} g_{i,n}\right|$, for $1\leq n \leq N$, ${\rm P}_1^{{\rm NOMA}} $ can be upper bounded as follows:
\begin{align} \label{bad}
{\rm P}_1^{{\rm NOMA}} 
\leq&  \left(\int^{\epsilon_1}_{0}f_{\left| g_{0,n} g_{i,n}\right|^2}(x)  
dx\right)^N \\\nonumber =& \left(1-2\epsilon_1^{\frac{1}{2}} K_1\left(2\epsilon_1^{\frac{1}{2}} \right)\right)^N \approx \left(-\epsilon_1\log \epsilon_1\right)^N \doteq \frac{1}{P_0^N},
\end{align}
 which shows that the full diversity gain $N$ is achievable. 
Although the upper bound in \eqref{bad} is useful for the study of the diversity gain, we note that it is a bound much looser than the upper bound shown in Lemma 
\ref{lemma0}, particularly for the case of large $N$.    ${\rm P}_1^{{\rm I-OMA}} $ can be similarly obtained as ${\rm P}_1^{{\rm NOMA}} $.
\vspace{-0.5em}

\subsection{Random Phase Shifting} The use of random phase shifting can avoid the requirement  of perfect phase adjustment   and reduce  the system overhead needed for acquiring CSI at the source. Recall   $\xi_N$    as shown in  \eqref{theta}, 
where $\theta_n$ is  randomly chosen for random phase shifting.

Different from the coherent phase shifting case, $\xi_N$   is a sum of  complex-valued random variables,   and the imaginary and real parts of $e^{-j\theta_{n}} g_{0,n} g_{i,n}$ are correlated,  which means that the CLT is not directly applicable. In the following, we will show that $\xi_N$ can be still approximated as a complex Gaussian random variable.   We first note that, for small $N$, $\xi_N$ is not complex Gaussian distributed, as can be seen from the following special case.
\begin{proposition}\label{proposition1}
For the special case of $N=1$, the pdf of the real and imaginary parts of $\xi_N$ is given by
\begin{align}
f_{{\rm Re}\{\xi_1\}}(x) = e^{-2|x|}.
\end{align}
\end{proposition}
\begin{proof}
See Appendix B. 
\end{proof}
%$f_G(x) = \frac{1}{\sqrt{N\pi}}e^{-\frac{x^2}{N}}$

 However, by increasing  $N$,   the Gaussian approximation becomes applicable to   $\xi_N$,  as shown in the following lemma.  

\begin{lemma} \label{lemma1}
When $N\rightarrow \infty$, $\xi_N$ can be approximated as a complex Gaussian random variable with zero mean and variance $N$, i.e.,
\begin{align}
\xi_N \rightarrow \mathcal{CN}\left(0,N\right).
\end{align}
\end{lemma}
\begin{proof}
See Appendix C. 
\end{proof}
Therefore, the outage probabilities achieved by IRS  with random phase shifting can be approximated as follows: 
\begin{align}\label{iomax}
{\rm P}_1^{{\rm NOMA}} &\approx 1-e^{-\frac{\epsilon_1}{N }},\quad \& \quad {\rm P}_1^{{\rm I-OMA}} \approx 1-e^{-\frac{\epsilon_2}{N }} ,
\end{align} 
which indicates that IRS transmission with random phase shifting realizes a diversity order of one.

{\bf Phase shift selection:} As shown in the next section, random phase shifting cannot effectively use the spatial degrees of freedom, although it can be implemented at low complexity.  A better tradeoff between performance and complexity can be realized by carrying out phase shift selection, as described in the following. The IRS uses $Q$   sets of random phase shifts to send  $Q$ pilot signals. ${\rm U}_1$ informs the source which set of phase shifts it  prefers.   The use of this simple phase selection can significantly improve the performance of IRS transmission, as shown in the next section.  We also note that the effective channel gains, $\xi_N$, obtained with different sets of random phases are not independent, which makes it difficult to develop analytical results for     the proposed phase selection scheme.  

%We note that the performance of random phase shifting can be further improved by applying opportunistic phase selection, which will be discussed in the next section.    

 \vspace{-0.5em}
\section{Numerical Studies}\label{section sim}
In this section, the performance of the three considered transmission  schemes is evaluated by computer simulations. Without loss of generality,  we choose $c_1^2=0.8$, $c_2^2=0.2$,  $d_2=d_r=20$ m, $d_{r1}=d_{12}=10$ m, $\alpha=4$, $R_1=1.8$ bit per channel use (BPCU). $P_1=P_2=P_r=P_s$. The noise power is $-7 0$ dBm.    IRS-NOMA with coherent phase shifting is studied first in Fig.~\ref{fig1}.  As can be observed  in the figure, conventional relaying can outperform the two IRS transmission schemes, particularly in the low SNR regime. This performance loss is    due to the fact that IRS transmission suffers from severe path loss, $d_r^{\alpha}d_{r1}^{\alpha}$, as previously  pointed out in  \cite{skyx5}. However, by increasing the transmission power or the number of reflecting elements on the IRS, the IRS schemes eventually outperform conventional relaying, where the performance of IRS-NOMA is always better than that of IRS-OMA. The accuracy of the   developed CLT approximation and the upper bound is also evaluated in the figure. As can be seen from the figure, in the low SNR regime, the CLT based approximation is   accurate, and the developed upper bound is more accurate in the high SNR regime.

\begin{figure}[!t]\centering \vspace{-1em}
    \epsfig{file=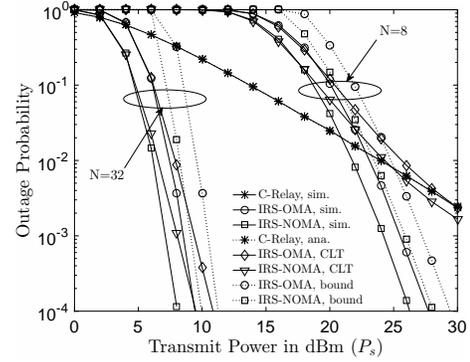, width=0.33\textwidth, clip=}\vspace{-1em}
\caption{ IRS-OMA and IRS-NOMA with the coherent phase shifting scheme vs  conventional relaying.   \vspace{-1.5em} }\label{fig1}\vspace{-0.5em}
\end{figure}

\begin{figure}[!htbp]\centering \vspace{-0.5em}
    \epsfig{file=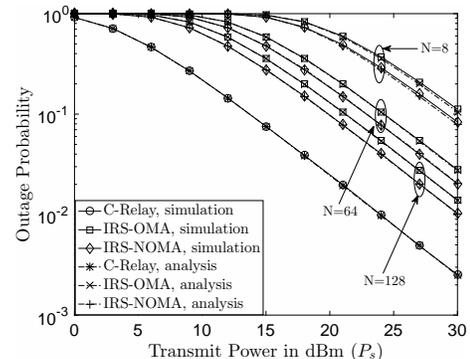, width=0.33\textwidth, clip=}\vspace{-1em}
\caption{ The impact of random phase shifting on IRS transmission.     \vspace{-1.5em} }\label{fig2}\vspace{-1em}
\end{figure}

\begin{figure}[!htbp]\centering \vspace{-0.5em}
    \epsfig{file=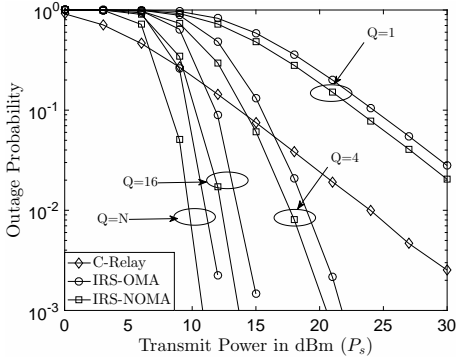, width=0.33\textwidth, clip=}\vspace{-1em}
\caption{ Impact of phase shift selection on IRS transmission.  $N=64$   \vspace{-1.5em} }\label{fig3} \vspace{-1em}
\end{figure}

In Fig. \ref{fig2}, the impact of random phase shifting on IRS transmission is investigated. Recall that the motivation to use random phase shifting is that it  can significantly reduce the system overhead and does not require complicated phase control mechanisms  compared to coherent phase shifting. However,  Fig. \ref{fig2}   shows that conventional relaying outperforms IRS transmission, although  increasing $N$ is helpful to  reduce the performance gap. The reason for this performance loss is due to the fact that random phase shifting cannot efficiently utilize the spatial degrees of freedom offered by IRS.  By implementing the proposed phase selection scheme,  the performance of IRS transmission can be significantly improved. As shown in Fig. \ref{fig3}, with  $Q=4$, IRS-NOMA can realize  a power reduction  of $10$ dBm  at an outage probability of $10^{-3}$, compared to conventional relaying. 
Figs. \ref{fig2} and \ref{fig3} also show that IRS-NOMA always outperforms IRS-OMA, which is consistent with Fig. \ref{fig1}. Furthermore, Fig. \ref{fig2}   demonstrates the accuracy of the approximation based on Lemma \ref{lemma1}.
\vspace{-0.5em}
\section{Conclusions}
In this letter, the impact of   two phase shifting designs on the performance of IRS-NOMA has been studied. Analytical results were developed to show that the two designs achieve different tradeoffs between system performance and complexity.  Simulation results were  provided to show the accuracy of the obtained analytical results and to compare IRS-NOMA to conventional relaying and IRS-OMA. For coherent phase shifting, the pdf of the effective channel gain, $\xi_N$, was evaluated  in this letter by using two types of approximations, and finding  an exact expression for the pdf is an important direction for future research. 

\vspace{-0.5em}
%%%%%%%%%%%%%%%%%%%%%%%%%%%%
\appendices 
\section{Proof for Lemma \ref{lemma0}}
The  outage probability achieved by IRS-OMA is given    as follows:
\begin{align}\label{iomax2}
{\rm P}_1^{{\rm I-OMA}}  
&=\underset{\sum^{N}_{i=1}y_i<\sqrt{\epsilon_1}}{\int\cdots\int} \prod_{i=1}^{N}f_{\left| g_{0,n} g_{i,n}\right|}(y_i) dy_i .
\end{align} 
 The fact that   the pdf of $\left| g_{0,n} g_{i,n}\right|$, $f_{\left| g_{0,n} g_{i,n}\right|}(y_i)$,  contains the Bessel function is the main reason why   the performance analysis is difficult. We note that  an upper bound on the Bessel function was provided in  \cite{Yangbessel} as follows:
\begin{align}
K_0(x)\leq \frac{\sqrt{\pi}e^{-x}}{\sqrt{2 x}}.
\end{align}
By using this upper bound on the Bessel function,  the pdf of $\left| g_{0,n} g_{i,n}\right|$ shown in \eqref{abs1} can be upper bounded as follows:
\begin{align}\label{abs2}
f_{\left| g_{0,n} g_{i,n}\right|}(y) \leq  4y  \frac{\sqrt{\pi}e^{-2y}}{\sqrt{4y}} = 2\pi^{\frac{1}{2}}y^{\frac{1}{2}} e^{-2y}\triangleq g(y).
\end{align}
Because of the simple expression of $g(y)$, an upper bound on the pdf of the sum, $\sum^{N}_{i=1}y_i$, can be obtained, as shown in the following. First, the Laplace transform of the upper bound   $g(y) $ can be obtained as follows:
\begin{align}
\mathcal{L}\left(g(y)\right) = &2\pi^{\frac{1}{2}}\int^{\infty}_{0}e^{-sy}y^{\frac{1}{2}} e^{-2y}dy=\frac{2\pi^{\frac{1}{2}}\Gamma\left(\frac{3}{2}\right)}{(s+2)^{\frac{3}{2}}}.
\end{align}

By using the fact that $y_i$ is i.i.d.,  the pdf of the sum, denoted by $f_{\sum^{N}_{i=1}y_i}(y)$, can be upper bounded as follows:
\begin{align}
f_{\sum^{N}_{i=1}y_i}(y)  \leq &\mathcal{L}^{-1}\left(\frac{2^N\pi^{\frac{N}{2}}\Gamma^N\left(\frac{3}{2}\right)}{(s+2)^{\frac{3N}{2}}}\right).
\end{align}

Assume that $N$ is an even number and let  $\bar{N}=\frac{N}{2}$.  An upper bound on the pdf of the sum  can be obtained as follows:
\begin{align}
f_{\sum^{N}_{i=1}y_i}(y)  \leq &\frac{2^N\pi^{\frac{N}{2}}\Gamma^N\left(\frac{3}{2}\right)}{ (3\bar{N}-1)!} y^{3\bar{N}-1}e^{-2y}.
\end{align}

Therefore, the outage probability achieved by IRS-OMA can be upper bounded as follows: 
\begin{align}\label{iomax41}
{\rm P}_1^{{\rm I-OMA}} 
&\leq  \int^{\sqrt{\epsilon_1}}_{0} \frac{2^N\pi^{\frac{N}{2}}\Gamma^N\left(\frac{3}{2}\right)}{ (3\bar{N}-1)!} y^{3\bar{N}-1}e^{-2y}dy\\\nonumber 
=&   \frac{2^N\pi^{\frac{N}{2}}\Gamma^N\left(\frac{3}{2}\right)}{ (3\bar{N}-1)!}  2^{-3\bar{N}} \gamma\left(3\bar{N}, 2\sqrt{\epsilon_1}\right).
\end{align} 
Thus,  the lemma is proved. 
\vspace{-1em}

\section{Proof for Proposition \ref{proposition1}}

Recall that $g_{0,n}$ and $g_{i,n}$ are complex Gaussian distributed  with zero mean and unit variance, i.e., $g_{0,n}\sim \mathcal{CN}(0,1)$ and $g_{i,n}\sim \mathcal{CN}(0,1)$. Without loss of generality, denote $g_{0,n}$ and $g_{i,n}$ by   $g_{0,n}=a_n+jb_n$ and $g_{i,n}=c_n+jd_n$, respectively, where $a_n$, $b_n$, $c_n$, and $d_n$ are i.i.d., and follow $\mathcal{N}\left(0, \frac{1}{2}\right)$.  Therefore,  $\xi_1$ can be written as follows:  
\begin{align}
\xi_1 &= (a_nc_n-b_nd_n)\cos \theta_{n}-(a_nd_n+b_nc_n)\sin \theta_{n}\\\nonumber &+j\left((a_nc_n-b_nd_n)\sin \theta_{n}+(a_nd_n+b_nc_n)\cos \theta_{n}\right).
\end{align}

We note that   the real and imaginary parts of $\xi_1$ are identically distributed, and hence in the following, we focus on the real part of $\xi_1$ which can be further written as follows:
\begin{align}
\rm{Re}\{\xi_1\} &=  a_n\tilde{c}_n-b_n\tilde{d}_n,
\end{align} 
where $\tilde{c}_n$ and $\tilde{d}_n$ are defined as follows:
\begin{align}\label{12z}
\begin{bmatrix}\tilde{c}_n\\ \tilde{d}_n  \end{bmatrix} = \begin{bmatrix} \cos \theta_{n} &-\sin \theta_{n} \\ \sin \theta_{n} &\cos \theta_{n} \end{bmatrix}\begin{bmatrix} {c}_n\\  {d}_n  \end{bmatrix}. 
\end{align} 
We further note that the transformation matrix in \eqref{12z} is a unitary matrix, and    $c_n$ and $d_n$  are i.i.d. Gaussian random variables. Therefore,   $\tilde{c}_n, \tilde{d}_n \sim \mathcal{N}\left(0, \frac{1}{2}\right)$ since a unitary transformation does not change the statistical property of Gaussian variables.

 Therefore, the cumulative distribution function (CDF) of $\rm{Re}\{\xi_1\}$ is given by
\begin{align}
{\rm P}(\rm{Re}\{\xi_1\}<y) =&{\rm P}(a_n\tilde{c}_n-b_n\tilde{d}_n<y)  \\\nonumber \overset{(i)}{=} &E_{\tilde{c}_n,\tilde{d}_n}\left\{ \int^{y}_{-\infty} \frac{1}{\sqrt{\pi (\tilde{c}_n^2+\tilde{d}_n^2)}}e^{-\frac{x^2}{\tilde{c}_n^2+\tilde{d}_n^2}}dx\right\}
\\\nonumber
\overset{(ii)}{=} & \int^{\infty}_0\int^{y}_{-\infty} \frac{1}{\sqrt{\pi z}}e^{-\frac{x^2}{z}}dx e^{-z}dz,
\end{align}
where step (i) follows the fact that $a_n\tilde{c}_n-b_n\tilde{d}_n$ is a sum of two i.i.d. Gaussian variables, $a_n$ and $b_n$, by treating $\tilde{c}_n$ and $-\tilde{d}_n$ as the weighting constants, and step (ii) follows by the fact that $z\triangleq \tilde{c}_n^2+\tilde{d}_n^2$ is exponentially distributed.

The CDF of $\rm{Re}\{\xi_1\}$ can be simplified as follows:
\begin{align}\nonumber
{\rm P}(\rm{Re}\{\xi_1\}<y) = &  \int^{y}_{-\infty}  \frac{1}{\sqrt{\pi }}\int^{\infty}_0 z^{-\frac{1}{2}}e^{-\frac{x^2}{z}-z}dzdx
\\ 
\overset{(iii)}{=} & \int^{y}_{-\infty}  e^{-2|x|}dx,
\end{align}
where step (iii) follows   \cite[(3.471.15)]{GRADSHTEYN}.  By finding the derivative of the CDF, the pdf in the proposition can be obtained and the proof is complete. 

\vspace{-1em}
\section{Proof for Lemma \ref{lemma1}}
The lemma can be proved by two steps as shown in the following two subsections.   
\vspace{-1.5em}
\subsection{${\rm Re}\{\xi_N\}$, ${\rm Im}\{\xi_N  \} \sim \mathcal{N}\left(0, \frac{N}{2}\right)$  for $N\rightarrow \infty$}
Since ${\rm Re}\{\xi_N\}$ and ${\rm Im}\{\xi_N  \}$ are identically distributed, we will focus on ${\rm Re}\{\xi_N\}$, without loss of generality. Recall that $
{\rm Re}\{\xi_N\} = \sum^{N}_{n=1}{\rm  Re}\{e^{-j\theta_{n}} g_{0,n} g_{i,n}\}$. 

Without directly using the CLT,   the approximation for the pdf of ${\rm Re}\{\xi_N\}$ can be straightforwardly obtained   as follows.
By   applying Proposition \ref{proposition1}, the characteristic function of the pdf of ${\rm Re}\{e^{-j\theta_{n}} g_{0,n} g_{i,n}\}  $ can be obtained as follows:
\begin{align}
\psi_{{\rm Re}\{e^{-j\theta_{n}} g_{0,n} g_{i,n}\} }(t)
  &=\int^{\infty}_{-\infty} e^{itx}f_{\xi_1}(x)dx  =\frac{4}{4+t^2}.
\end{align} 
 By using the fact that  ${\rm Re}\{e^{-j\theta_{n}} g_{0,n} g_{i,n}\} $ is independent from ${\rm Re}\{e^{-j\theta_{m}} g_{0,m} g_{i,m}\} $ for $n\neq m$, the characteristic function of the pdf of ${\rm Re}\{\xi_N\} $ can be obtained as follows:
\begin{align}
\psi_{{\rm Re}\{\xi_N\} }(t) = & \frac{4^N}{\left(4+t^2\right)^N} \underset{N\rightarrow \infty}{\rightarrow} e^{-\frac{Nt^2}{4}},
\end{align}
where the approximation follows by applying the limit of the exponential function. Therefore,   ${\rm Re}\{\xi_N\}$ can be approximated as a Gaussian random variable since     $e^{-\frac{Nt^2}{4}}$ is    the characteristic function of a Gaussian random variable.

\vspace{-1.5em}
\subsection{${\rm Re}\{\xi_N\}$ and ${\rm Im}\{\xi_N\}$ are Independent for $N\rightarrow \infty$}
 
We have proved that ${\rm Re}\{\xi_N\},{\rm Im}\{\xi_N\}\sim \mathcal{N}\left(0, \frac{N}{2}\right)$, for $N\rightarrow \infty$. Therefore, the independence between the two random variables can be proved by showing them to be jointly Gaussian distributed and also uncorrelated. 

In order to show that    ${\rm Re}\{\xi_N\}$ and ${\rm Im}\{\xi_N\}$ are jointly Gaussian distributed, we first build  an arbitrary  linear combination  of ${\rm Re}\{\xi_N\}$ and ${\rm Im}\{\xi_N\}$ with $\beta_1$ and $\beta_2$ as follows:
\begin{align}\label{reim}
&\beta_1 {\rm Re}\{\xi_N\} +\beta_2{\rm Im}\{\xi_N\} \\\nonumber =&  \sum_{n=1}^N a_n(\beta_1\tilde{c}_n+\beta_2\tilde{d}_n)-b_n(\beta_1\tilde{d}_n-\beta_2\tilde{c}_n).
\end{align}
Note that $(\beta_1\tilde{c}_n+\beta_2\tilde{d}_n)$ and $(\beta_1\tilde{d}_n-\beta_2\tilde{c}_n)$ are independent and identically Gaussian distributed since they are constructed  from $\tilde{c}_n$ and $\tilde{d}_n$ with two orthogonal coefficient vectors $\begin{bmatrix}\beta_1&\beta_2\end{bmatrix}^T$ and $\begin{bmatrix}-\beta_2&\beta_1\end{bmatrix}^T$. By following the steps in the previous subsection, it is straightforward to show that the linear combination is also Gaussian distributed, which means that $\rm{Re}\{\xi_N\}$ and $\rm{Im}\{\xi_N\}$ are jointly Gaussian distributed. 

The correlation between ${\rm Re}\{\xi_N\}$ and ${\rm Im}\{\xi_N\}$ is given by
\begin{align}
&\mathcal{E}\left\{ {\rm Re}\{\mathbf{g}_i^H\boldsymbol \Theta_p \mathbf{g}_0\}  {\rm Im}\{\mathbf{g}_i^H\boldsymbol \Theta_p \mathbf{g}_0\}\right\} \\\nonumber =&\mathcal{E}\left\{\left( \sum_{n=1}^N a_n\tilde{c}_n-b_n\tilde{d}_n \right)\left( \sum_{n=1}^N a_n\tilde{d}_n+b_n\tilde{c}_n\right)\right\}
= 0,
\end{align}
which follows   the fact that $a_n$, $b_n$, $\tilde{c}_n$ and $\tilde{d}_n$ are i.i.d.    with zero mean, where $\mathcal{E}\{\cdot\}$ denotes the expectation. 
Therefore, the independence between ${\rm Re}\{\xi_N\}$ and ${\rm Im}\{\xi_N\}$ is proved. 
\vspace{-0.9em}
   \bibliographystyle{IEEEtran}
\bibliography{IEEEfull,trasfer}
 
   \end{document}